\def\draft{0}
\def\llncs{0}
\def\anon{0}

\ifnum\llncs=0
\documentclass[11pt]{article}
\usepackage{booktabs}
\usepackage{fullpage}
\usepackage{amsmath,amsfonts,amsthm,mathrsfs,xspace,graphicx}
\usepackage{mathpazo}
\usepackage{times}
\else
\documentclass[orivec,runningheads,envcountsame,envcountsect]{llncs}
\usepackage{amssymb,amsmath,cite,url}
\fi

\usepackage[backref,colorlinks,citecolor=blue,bookmarks=true]{hyperref}
\usepackage{endnotes}
\usepackage{color}
\usepackage{float}
\usepackage{xcolor}
\usepackage{mdframed}
\usepackage{bbm}
\usepackage{suffix} %
\usepackage{tabularx}
\usepackage{makecell}
\usepackage{amssymb,latexsym}
\usepackage{IEEEtrantools}
\usepackage[capitalize]{cleveref}
\usepackage{enumitem}
\usepackage{tikz}
\usepackage{tikz-cd}
\usepackage{multirow}
\usepackage[section]{placeins}

\mdfdefinestyle{figstyle}{ %
  linecolor=black!7, %
  backgroundcolor=black!7, %
  innertopmargin=10pt, %
  innerleftmargin=25pt, %
  innerrightmargin=25pt, %
  innerbottommargin=10pt %
}

\ifnum\llncs=0

\newtheorem{theorem}{Theorem}[section]
\newtheorem{assumption}{Assumption}

\newtheorem{proposition}[theorem]{Proposition}

\newtheorem{lemma}[theorem]{Lemma}

\newtheorem{nclaim}[theorem]{Claim}

\newtheorem{corollary}[theorem]{Corollary}

\theoremstyle{remark}

\theoremstyle{definition}
\newtheorem{definition}[theorem]{Definition}

\else
\spnewtheorem{assumption}{Assumption}{\bfseries}{\itshape}
\spnewtheorem{nclaim}[theorem]{Claim}{\bfseries}{\itshape}
\fi

\ifnum\draft=1
\def\ShowAuthNotes{1}
\else
\def\ShowAuthNotes{0}
\fi

\ifnum\ShowAuthNotes=1
\newcommand{\authnote}[3]{\textcolor{#3}{[{\footnotesize {\bf #1:} { {#2}}}]}}
\else
\newcommand{\authnote}[3]{}
\fi

\newcommand{\beq}{\begin{eqnarray}}
\newcommand{\eeq}{\end{eqnarray}}

\newcommand{\ket}[1]{|#1\rangle}
\newcommand{\bra}[1]{\langle#1|}

\newcommand{\Id}{\ensuremath{\mathop{\rm Id}\nolimits}}
\newcommand{\Es}[1]{\ensuremath{\mathop{\textsc{E}}}_{#1}}

\newcommand{\device}{\mathfrak{D}}
\newcommand{\negl}{\ensuremath{\mathop{\rm negl}}}

\newcommand{\reg}[1]{{\textsf{#1}}}

\newcommand{\C}{\ensuremath{\mathbb{C}}}
\newcommand{\N}{\ensuremath{\mathbb{N}}}

\newcommand{\R}{\ensuremath{\mathbb{R}}}

\newcommand{\mH}{\mathcal{H}}

\newcommand{\mX}{\mathcal{X}}
\newcommand{\mA}{\mathcal{A}}
\newcommand{\mK}{\mathcal{K}}

\newcommand{\mS}{\mathcal{S}}
\newcommand{\mT}{\mathcal{T}}
\newcommand{\mY}{\mathcal{Y}}

\newcommand{\rP}{\reg{P}}

\newcommand{\rE}{\reg{E}}

\newcommand{\eps}{\varepsilon}

\newcommand{\trans}{\ensuremath{\texttt{trans}}}
\newcommand{\rand}{\ensuremath{\texttt{rand}}}
\newcommand{\flag}{\ensuremath{\texttt{flag}}}
\newcommand{\acc}{\ensuremath{\texttt{acc}}}
\newcommand{\rej}{\ensuremath{\texttt{rej}}}
\newcommand{\cont}{\ensuremath{\texttt{cont}}}

\newcommand{\Gen}{\mathrm{Gen}}
\newcommand{\Inv}{\mathrm{Inv}}
\newcommand{\Enc}{\mathrm{Enc}}
\newcommand{\Dec}{\mathrm{Dec}}
\newcommand{\Eval}{\mathrm{Eval}}

\newenvironment{gamespec}{
  \begin{mdframed}[style=figstyle]}{
  \end{mdframed}}

\begin{document}

\ifnum\llncs=1
\author{Zvika Brakerski\inst{1}
	\and Alexandru Gheorghiu\inst{2}
	\and Gregory D.\ Kahanamoku-Meyer\inst{3,4}
	\and Eitan Porat\inst{1}
	\and Thomas Vidick\inst{1,5}}
\institute{Weizmann Institute of Science, Israel \and Chalmers University of Technology, Sweden \and Lawrence Berkeley National Laboratory, Berkeley, California, USA \and University of California, Berkeley, California, USA \and California Institute of Technology}
\authorrunning{Brakerski, Gheorghiu, Kahanamoku-Meyer, Porat and Vidick}
\date{}
\else
\ifnum\anon=1
\author{}
\else
\author{Zvika Brakerski\thanks{Weizmann Institute of Science, Israel, \texttt{zvika.brakerski@weizmann.ac.il}. Supported by the Israel Science Foundation (Grant No.\ 3426/21), and by the European Union Horizon 2020 Research and Innovation Program via ERC Project REACT (Grant 756482).}
\and Alexandru Gheorghiu\thanks{Chalmers University of Technology, Sweden. \texttt{alexandru.gheorghiu@chalmers.se}. Supported by the Knut and Alice Wallenberg Foundation through the Wallenberg Centre for Quantum Technology (WACQT).}
\and Gregory D. Kahanamoku-Meyer\thanks{Lawrence Berkeley National Laboratory \& University of California, Berkeley, California, USA. \texttt{gkm@berkeley.edu}. Supported by the U.S. Department of Energy, Office of Science, Office of Advanced Scientific Computing Research, under the Accelerated Research in Quantum Computing (ARQC) program.}
\and Eitan Porat\thanks{Weizmann Institute of Science.}
\and Thomas Vidick\thanks{California Institute of Technology and Weizmann Institute of Science. Email: \texttt{thomas.vidick@weizmann.ac.il}. Supported by a grant from the Simons Foundation (828076, TV), MURI Grant FA9550-18-1-0161, AFOSR Grant FA9550-21-S-0001, and a research grant from the Center for New Scientists at the Weizmann Institute of Science.}
}
\fi

\ifnum\draft=1
\date{\today}
\else
\date{}
\fi

\fi

\sloppy %

\title{Simple Tests of Quantumness Also Certify Qubits
\ifnum\llncs=1
\thanks{The full version of this paper can be found at \url{https://arxiv.org/abs/2303.01293}.}
\fi
}

\maketitle

\begin{abstract}

A test of quantumness is a protocol that allows a classical verifier to certify (only) that a prover is not classical. 
We show that tests of quantumness that follow a certain template, which captures recent proposals such as~\cite{KCVY2022,kalai2022quantum}, can in fact do much more.
Namely, the same protocols can be used for \emph{certifying a qubit}, a building-block that stands at the heart of applications such as certifiable randomness and classical delegation of quantum computation.

Certifying qubits was previously only known to be possible based on families of post-quantum trapdoor claw-free functions (TCF) with an advanced ``adaptive hardcore bit'' property, which have only been constructed based on the hardness of the Learning with Errors problem~\cite{brakerski2021cryptographic} and recently isogeny-based group actions~\cite{alamati2023candidate}. 
Our framework allows certification of qubits based only on the existence of post-quantum TCF, without the adaptive hardcore bit property, or on quantum fully homomorphic encryption. These can be instantiated, for example, from Ring Learning with Errors.
This has the potential to improve the efficiency of qubit certification and derived functionalities.

On the technical side, we show that the \emph{quantum soundness} of any such protocol can be reduced to proving a bound on a simple algorithmic task: informally, answering ``two challenges simultaneously'' in the protocol. 
Our reduction formalizes the intuition that these protocols demonstrate quantumness by leveraging the impossibility of rewinding a general quantum prover. 
This allows us to prove tight bounds on the quantum soundness of~\cite{KCVY2022} and~\cite{kalai2022quantum}, showing that no quantum polynomial-time prover can succeed with probability larger than $\cos^2 \frac{\pi}{8}\approx 0.853$. 
Previously, only an upper bound on the success probability of classical provers, and a lower bound on the success probability of quantum provers, were known. 
We then extend this proof of quantum soundness to show that provers that approach the quantum soundness bound must perform almost anti-commuting measurements.
This certifies that the prover holds a qubit.
\end{abstract}

\newpage

\section{Introduction}

A \emph{cryptographic test of quantumness}\footnote{Also sometimes referred to as \emph{proof of quantumness}~\cite{brakerski2020simpler}.} is an interactive protocol allowing a classical polynomial-time \emph{verifier} to determine, with high confidence, that a (possibly quantum) polynomial-time \emph{prover} with which the verifier is interacting is non-classical. More precisely, it should be the case that there exists a quantum polynomial-time prover that succeeds with high probability in the protocol (\emph{quantum completeness}) whereas no classical polynomial-time prover can succeed with comparable probability (\emph{classical soundness}). Ideally, the former statement should hold with a ``simple'' quantum prover (i.e. one performing small quantum circuits), and the latter statement should hold based on the weakest possible cryptographic assumption.

A simple example of a test of quantumness consists in asking the prover to factor a large integer chosen by the verifier. While this proposal satisfies both completeness and soundness at the coarsest level, it suffers from two important limitations. The first is that demonstrating success for the quantum prover requires executing Shor's quantum algorithm for factoring~\cite{shor1994algorithms}, which although quantum polynomial-time requires a large circuit to be executed (roughly, $\tilde{O}(\lambda)$ qubits and $\tilde{O}(\lambda^2)$ circuit size,  where $\lambda$ is the bit-length of the integer).

The second, even more significant limitation, is that it is unclear how to build any other interesting cryptographic primitive on top of such a protocol. Indeed, beyond the near-term demonstration of a quantum advantage the impetus for studying tests of quantumness comes from their potential use as building blocks towards more complex protocols in quantum cryptography, such as protocols for certified randomness~\cite{brakerski2021cryptographic, mahadev2022efficient}, for testing that a quantum prover is able to coherently manipulate a certain number of qubits (``tests for quantum space'')~\cite{fu2022computational, gheorghiu2022quantum}, or for classical delegation of quantum computations~\cite{mahadev2018classical, gheorghiu2019computationally, zhang2022classical}. Currently, the only known protocols for such tasks all have at their heart a simple test of quantumness. In fact, to the best of our knowledge for all known protocols the test of quantumness is the same --- it is the test introduced in~\cite{brakerski2021cryptographic}.\footnote{The one exception is the recent test of quantumness by Yamakawa and Zhandry~\cite{yamakawa2022verifiable}, which also achieves certifiable randomness generation, albeit by relying on a conjecture of Aaronson and Ambainis~\cite{aaronson2014need}. We exclude ``quantum supremacy'' demonstrations such as~\cite{aaronson2011computational,GoogleSupremacy} because (i) they are not efficiently verifiable, and (ii) except for a single exception~\cite{bassirian2021certified}, they are not known to lead to any interesting cryptographic task.} The soundness of this test was originally shown based on the post-quantum hardness of the Learning With Errors (LWE) problem~\cite{regev2009lattices}, and more recently the same test (in fact, the same underlying assumption of post-quantum TCF with adaptive hardcore bit) was also shown sound based on isogeny-based group actions~\cite{alamati2023candidate}, an advanced post-quantum cryptographic assumption. 

The pervasiveness of the test of quantumness from~\cite{brakerski2021cryptographic} for advanced cryptographic applications begs the question --- can new tests of quantumness be found, that preserve the versatility of the previous test and its potential for applications but simplify, or at least diversify, the range of assumptions on which the test is based? Besides its intrinsic interest, answering this question may lead to more versatile as well as more prover-efficient tests; the test~\cite{brakerski2021cryptographic} requires quite an aggressive setting of parameters for LWE which hampers its practical applicability (see however the proof of principle demonstration in~\cite{zhu2021interactive}).

Recently there has been progress on this question in two different directions. Firstly, in~\cite{brakerski2020simpler} the authors introduce a test that is sound in the (quantum) random oracle model, assuming only the existence of trapdoor claw-free functions (defined later). Going even further~\cite{yamakawa2022verifiable} show a \emph{non-interactive} test of quantumness in the random oracle model. Secondly, very recently there have been two new proposals in the standard model. The test from~\cite{KCVY2022} is a 6-message protocol whose soundness can be based on any family of trapdoor claw-free functions. The test from~\cite{kalai2022quantum} is a 4-message protocol, that can be instantiated based on any two-player nonlocal game with a quantum advantage, and whose soundness relies on the existence of a quantum homomorphic encryption scheme, with some specific properties (which are known to hold for e.g.\ the scheme in~\cite{mahadev2020classical}).

However, these latter tests of quantumness suffer from an important limitation: prior to our work, it was not known how they could be expanded into a test for certified randomness, a test for a qubit (defined informally below), let alone a delegation protocol. This is simply because it was unknown how to analyze the behavior of malicious quantum provers in the protocol! For these protocols, only quantum \emph{completeness} was known, but it was not known if the proposed honest strategy for the quantum prover is optimal. For example, for the proposals in~\cite{KCVY2022,kalai2022quantum} it was not even known if the success probability demonstrated by the honest quantum prover (in both cases, $\cos^2\frac{\pi}{8}\approx 0.853\ldots$) was the optimal success probability --- it was left open if there could be a quantum prover that succeeds with probability $1$.

The main technical difficulty is that for showing soundness against classical adversaries, one can use classical rewinding-type arguments. This is no longer possible against quantum provers: quantum rewinding is notoriously delicate, and while some general results have started to appear (e.g.\ \cite{CMSZ21}) they do not apply to the above tests.
In fact, in some sense they \emph{should not} apply, because the goal of the test is to demonstrate a quantum advantage through the quantum demonstration of a \emph{non-rewindable} task.

\subsection{Our Results}

We break the ``quantum soundness barrier'' by showing that $\cos^2 \frac{\pi}{8}$ is indeed the optimal success probability for the protocols introduced in~\cite{KCVY2022,kalai2022quantum}. In fact, we show a general bound that applies to a broad class of protocols encompassing the above two.
 Furthermore, we are able to show that the measurement operators used by any near-optimal prover must satisfy a form of approximate anti-commutation. Such anti-commutation is a key signature of quantumness, which is informally known as a ``test for a qubit'' in the literature (see e.g.~\cite[Lecture 2]{vidick2020course} for more on this). A very similar statement forms the basis for the certified randomness protocol in~\cite{brakerski2021cryptographic}. Our results thus open the door for making use of the protocols from~\cite{KCVY2022,kalai2022quantum} for the same applications as the one from~\cite{brakerski2021cryptographic}, from certified randomness to delegated computation~\cite{mahadev2018classical,gheorghiu2019computationally}. The main advantage of our protocol is that its analysis does not require the infamous ``adaptive hardcore bit'' property used in~\cite{brakerski2021cryptographic}. As a result, we can instantiate it with any family of trapdoor claw-free functions, such as can be constructed from e.g.\ the ring-LWE problem, leading to more efficient protocols and better parameters.\footnote{We recall that the use of adaptive hardcore bits in \cite{brakerski2021cryptographic} led to a significant degradation in the parameters for LWE that could be used.}

Before describing our results, and their proof, in more detail, we give a high-level overview of our ``template protocol'' (see Figure~\ref{fig:template} for a summary), which specifies a general format for a test of quantumness to which our results apply. 
Our template is divided into two phases. The first phase, Phase A, is a setup phase in which the verifier and prover exchange classical information that, informally, enables the quantum prover to create the ``right'' initial state. 
For example, for the protocol from~\cite{KCVY2022} the result of Phase A is that the honest prover has returned to the verifier a string $y$, which has two preimages $x_0$ and $x_1$ under a claw-free function $f$, as well as an $n$-bit string $d$, and furthermore the prover holds the quantum state
\begin{equation*}
 \frac{1}{\sqrt{2}}\big( \ket{r\cdot x_0} +  (-1)^{d\cdot(x_0\oplus x_1)} \ket{r\cdot x_1}\big)\;,
\end{equation*}
for some uniformly random string $r$ that is communicated to the prover at some point during the phase.
(Note that the protocol of \cite{brakerski2021cryptographic} can also be cast in this language, although the equivalent of Phase B has a different structure so that protocol does not fully fall under the template we describe here.)
Note that a ``phase'' is allowed to take place over multiple rounds of interaction. 
Phase A may also incorporate a test executed by the verifier, in which case the verifier is allowed to end the protocol outright with a flag signifying acceptance or rejection of the prover.
(The ``preimage test'' in \cite{brakerski2021cryptographic,KCVY2022} is an example of such a test that occurs in Phase A).

Assuming no such flag has been raised, the protocol proceeds to Phase B. This phase is very simple:
the prover is sent a uniformly random challenge bit $m\in \{0,1\}$ and required to respond with a bit $b\in \{0,1\}$.
 The verifier classically computes a correct value $\hat{c}_m = \hat{c}_m(\rand,\trans)$, where $\rand$ are the verifier's private random bits and $\trans$ the classical communication transcript from Phase A, and the verifier accepts if and only if $(-1)^b = \hat{c}_m$.

Our main result is a reduction from classical, respectively quantum, soundness of the protocol template to a specific guessing task. Informally, we show the following.

\begin{theorem}[Main theorem, informal]
Suppose that no classical (resp. quantum) polynomial-time algorithm may produce a guess for the parity $\hat{c}_0\cdot \hat{c}_1$ that is correct with a non-negligible advantage.\footnote{Here, $\hat{c}_0$ and $\hat{c}_1$ are defined conditional on a transcript for Phase A of the protocol template. We refer to Theorem~\ref{thm:main} for a more precise formulation.} Then no classical (resp. quantum) polynomial-time prover can succeed in the protocol template with probability larger than $\frac{3}{4}$ (resp.\ $\cos^2 \frac{\pi}{8}$) by more than a negligible amount.
\end{theorem}

The strength of this result is that it reduces an a priori complex task --- bounding the success probability of a classical or quantum prover in an interactive protocol --- to showing limitations on adversaries in a much simpler, non-interactive task. 
In particular, as we will see from the examples (Section~\ref{sec:applications}), in virtually all known cases the hardness of the parity-predicting task considered in the theorem is essentially immediate from the computational hardness assumption that underlies the protocol. 
For example, in the protocol from~\cite{KCVY2022} the product $\hat{c}_0\cdot \hat{c}_1$ turns out to equal $r\cdot (x_0\oplus x_1)$. 
A prover who can predict this quantity with advantage $\eps$ for a uniformly random $r$ can, via the Golreich-Levin theorem, recover $x_0\oplus x_1$ with a related advantage. 
Since Phase A of this version of the protocol verifies that the prover also knows either $x_0$ or $x_1$, we conclude that the prover is able to recover a claw $(x_0,x_1)$, violating the claw-free property of the underlying function family.

The conceptually most appealing feature of our theorem is that it ``explains'' the bounds $\frac{3}{4}$ and $\cos^2 \frac{\pi}{8}$ that had been observed without justification for previous tests of quantumness, e.g.~\cite{KCVY2022,kalai2022quantum}. In particular, as already mentioned $\cos^2 \frac{\pi}{8}$ was established as a lower bound, but not an upper bound, on the maximum success probability of a quantum prover in those works. As such our result can be interpreted as a form of computational ``quantum non-rewinding'' result, which precisely establishes the extent to which a classical or quantum procedure may produce guesses for two quantities whose parity is known to be hard to predict.

Going beyond quantum soundness, we also establish that any quantum prover which succeeds close to the optimum probability must do so by using two binary measurements in Phase B of the protocol that are close to maximally anti-commuting measurements. This is what is generally known as a ``test for a qubit'' in the quantum information literature. Informally,

\begin{theorem}[Qubit test, informal]
Suppose that a quantum polynomial-time prover succeeds in the protocol template with probability at least $\cos^2 \frac{\pi}{8} - \eps$. Let $S_0$ and $S_1$ be the two binary observables associated with the prover's measurements on a Phase B challenge of $m=0$ and $m=1$ respectively. Then $S_0$ and $S_1$ are within distance $O(\sqrt{\eps})$ of a pair of perfectly anti-commuting measurements.\footnote{Here, ``distance'' should be measured using the appropriate norm. We use the standard ``state-dependent norm'' from self-testing. See Theorem~\ref{thm:qubit} for the precise formulation.}
\end{theorem}

This result is formalized in Theorem~\ref{thm:qubit}. 
While more technical, the statement will be familiar to researchers in the area of self-testing, and it is well-known to have powerful consequences. 
In~\cite{brakerski2021cryptographic} a similar statement is used to obtain certified randomness, and this has been expanded in~\cite{gheorghiu2019computationally} to obtain verifiable classical delegation of quantum computation. 
In a work of Merkulov and Arnon-Friedman~\cite{AM23} the connection with certified randomness is made explicit. 
They show that the specific bound on the anti-commutator obtained in our qubit test implies precise quantitative bounds on the randomness generated in a single execution of our (and other) protocols, as well as on the accumulation of randomness across many sequential executions. 

\paragraph{Related work.} While writing our results, we learned that Natarajan and Zhang~\cite{NZ23} had independently obtained directly related, yet strictly incomparable, results. Natarajan and Zhang focus on the protocol from~\cite{kalai2022quantum} (whereas our result applies more generically), when specialized to the CHSH nonlocal game, and establish its quantum soundness as well as the property of ``test for a qubit.'' This part of the results is common to both our works. However, they go further by showing that the test for a qubit can be leveraged to implement a complete protocol for verifiable delegation of quantum computations, an application which we do not investigate (though we expect our qubit test to also yield a quantum verification protocol by making use of the history state construction as in~\cite{mahadev2018classical}).

\subsection{Technical Overview}

Beyond its conceptual clarity, another appealing feature of our main result is that its proof is simple! We sketch the argument here. Let's start by observing that the statement for classical soundness is almost immediate. This is because a prover who is able to predict $\hat{c}_0$ with probability $p_0$, and $\hat{c}_1$ with probability $p_1$, such that $\frac{1}{2}(p_0+p_1)=\frac{3}{4}+\eps$, must, by a union bound, be able to predict both $\hat{c}_0$ and $\hat{c}_1$, and hence their product, with probability at least $\frac{1}{2}+2\eps$. If the latter is hard, then the prover cannot exist.

Now, let us think about the quantum case. Here the argument is more delicate, because the ``union bound'' does not apply. If it was known that the quantum prover responds correctly on at least one of the challenges with probability close to $1$, then we could use a tool such as the gentle measurement lemma (see e.g.~\cite[Lemma 9.4.1]{wilde2011classical}) to perform the rewinding. However, in general this will not be the case --- and indeed, it \emph{cannot} be the case, since we expect that there should be a quantum advantage!

The central question is thus the following: what is the smallest possible probability $p$ such that a quantum prover who can predict either $\hat{c}_0$, or $\hat{c}_1$, with overall probability $p$ (on average over the choice of a uniform $m\in \{0,1\}$), can also predict $\hat{c}_0\cdot \hat{c}_1$ with non-negligible advantage? We show that the answer is the famous probability $\cos^2 \frac{\pi}{8}\approx 0.853$. This probability already appears in the analysis of the nonlocal game CHSH,\footnote{CHSH is a well known two-prover protocol, which allows to certify quantum correlations between two non-communicating provers.} and it is not a coincidence. We explain why. Recall that in the CHSH game, two isolated provers Alice and Bob are given uniformly generated inputs $x,y \in \{0,1\}$ respectively, and are tasked with generating bits $a,b\in\{0,1\}$ respectively, where
\begin{equation}\label{eq:chsh}
a\oplus b \,=\, x\wedge y
\end{equation}
Observe that this condition means that, on input $y=0$ we must have $b_0 = a$, whereas on input $y=1$ we must have $b_1 = a\oplus x$. Therefore, Bob's task is to find two measurement operators whose outcomes satisfy $b_0\oplus b_1=x$. But $x$ is Alice's input, which Bob has no information at all about! The question of finding an optimal strategy in the CHSH game is then reduced to finding measurements for Bob that lead to the highest probability of success for~\eqref{eq:chsh}, while knowing that the parity $b_0\oplus b_1$ is a bit that is information-theoretically impossible for Bob to predict. It then immediately follows --- using the same argument as above --- that the maximum success probability for a classical Bob is $\frac{3}{4}$. And it is also known, albeit harder to prove, that the maximum quantum success probability is $\cos^2 \frac{\pi}{8}$, and that this can only be achieved using a pair of anti-commuting measurements for Bob.

Working out the actual result in our setting requires a bit more work. This is because the situation is not completely analogous to that of CHSH; in particular the state on which Bob (here, the quantum prover in Phase B of the protocol template) makes his measurement is the result of an interaction with the verifier, not of a measurement by Alice on some prior entangled state. While the technical setup differs, ultimately we are able to apply similar tools to those applied in quantum information theory for the analysis of non-local games. In particular, we make a careful use of Jordan's lemma to reduce the analysis to a $2$-dimensional problem (for this it is crucial that the verifier's challenge in the last round consists of a single bit). Having reduced the problem to two dimensions we quantify the tension between the tasks of succeeding in the protocol, and the potential of the quantum prover for predicting the parity with $\hat{c}_0\cdot \hat{c}_1$. Carefully working out this tension leads to the optimal quantitative tradeoffs that are expressed in our main theorem.

\subsection{Open Questions}

Our template protocol does not capture all known tests of quantumness. Two notable exceptions are the test of quantumness by Brakerski et al.~\cite{brakerski2021cryptographic} and the ones that operate in the random oracle model~\cite{brakerski2020simpler,yamakawa2022verifiable}. It may seem surprising that the test from~\cite{brakerski2021cryptographic} does not fit our framework, as indeed it is quite similar, though more demanding cryptographically, to the test from~\cite{KCVY2022}. It is possible that a small variation on the test could be made to fit our template, but we do not investigate this. Regarding the test from~\cite{yamakawa2022verifiable}, it operates in the random oracle model, which we did not attempt to incorporate in our framework. More importantly, it is non-interactive, which makes it unclear how our ideas could be used.

A test of quantumness that we believe should fall within a modified version of our framework is the application of the compiler from~\cite{kalai2022quantum} to the Magic Square game.
The Magic Square game is a nonlocal game which has the advantage that the optimal quantum winning probability is exactly $1$. 
This could lead to a test of quantumness with quantum completeness $1$, which is convenient for applications.
We leave this question open for future work.

By formalizing the common structure underlying many simple tests of quantumness,  our results suggest a hierarchy of ``capabilities'', that builds from a test of quantumness based on the non-rewinding property of quantum systems, to a test for a qubit, followed potentially by tests for certified randomness and delegated computation. 
An interesting conceptual question is to determine what is the minimal basis for achieving these capabilities, and whether the advanced ones can always, or almost always, be reduced to the more elementary ones, as seems to be the case in our framework. 
A specific direction that would be worthwhile investigating is whether certified randomness can be ``accumulated'' in a generic fashion from the family of protocols that we consider here.

\ifnum\anon=0
\paragraph{Acknowledgments.}
We thank Ilya Merkulov and Rotem Arnon-Friedman for discussions in the early stages of this work. 

\ifnum\llncs=1
The first author is supported by the Israel Science Foundation (Grant No.\ 3426/21), and by the European Union Horizon 2020 Research and Innovation Program via ERC Project REACT (Grant 756482). The second author is supported by the Knut and Alice Wallenberg Foundation through the Wallenberg Centre for Quantum Technology (WACQT). The third author is supported by the U.S. Department of Energy, Office of Science, Office of Advanced Scientific Computing Research, under the Accelerated Research in Quantum Computing (ARQC) program. The fifth author is supported by a grant from the Simons Foundation (828076, TV), MURI Grant FA9550-18-1-0161, AFOSR Grant FA9550-21-S-0001, and a research grant from the Center for New Scientists at the Weizmann Institute of Science.
\fi
\fi

\section{Preliminaries}

\subsection{Notation}

We use $\negl(\lambda)$ to denote any negligible function of $\lambda$, i.e.\ a function $f:\N\to\R_+$ such that $f(\lambda)p(\lambda)\to_{\lambda\to\infty} 0$ for all polynomials $p$. Given two strings $r_0,r_1$ we write $r_0\|r_1$ for their concatenation.

\subsection{Quantum Goldreich-Levin}

\noindent We state a quantum version of the Goldreich-Levin theorem, which is taken from~\cite{qgllemma}.

\begin{definition}\label{def:gl}
A \emph{quantum inner product} query (with bias $\varepsilon$) is a unitary transformation $U_{\mathrm{IP}}$ together with an auxiliary $m$-qubit quantum state $\ket{\psi}$ on $n+m+t$ qubits, or its inverse $U_{\mathrm{IP}}^{\dagger}$, such that $U_{\mathrm{IP}}$ satisfies the following two properties:
\begin{enumerate}
    \item There is a string $a\in \{0,1\}^n$ such that if $x\in\{0,1\}^{n}$ is chosen randomly according to the uniform distribution and the last qubit of $U_{\mathrm{IP}}\ket{x}\ket{\psi}\ket{0^{t}}$ is measured, the value $w\in\{0,1\}$ obtained is such that $\Pr(w=a\cdot x)\geq \frac{1}{2} + \varepsilon$.
    \item For any $x\in\{0,1\}^{n}$ and $y\in\{0,1\}^{t}$, the state of the first $n$ qubits of $U_{\mathrm{IP}}\ket{x}\ket{\psi}\ket{y}$ is $x$.
\end{enumerate}
\end{definition}

\begin{theorem}\label{thm:qgl}
There exists a quantum algorithm that returns the string $a$ with probability greater than or equal to $4\varepsilon^{2}$ using a $U_{\mathrm{IP}}$ query and a $U_{\mathrm{IP}}^{\dagger}$ query. The number of auxilliary qubit operation used by this procedure is $O(1)$.
\end{theorem}

\subsection{Jordan's lemma}

\begin{lemma}\label{lem:jordan}
Let $Q_0$ and $Q_1$ be two orthogonal projections on a (finite-dimensional) Hilbert space $\mH$. Then there is a decomposition $\mH=\oplus_i \mH_i$ where for each $i$, $\mH_i$ has dimension at most $2$ and furthermore, the decomposition is stabilized by both $Q_0$ and $Q_1$. In particular, whenever $\dim(\mH_i)=2$ we can find a basis of $\mH_i$ in which
\begin{align}
Q_0 \,=\,\begin{pmatrix} 1 & 0 \\ 0 & 0\end{pmatrix}\qquad\text{and}\qquad Q_1 = \begin{pmatrix} c_i^2 & c_is_i \\ c_is_i & s_i^2 \end{pmatrix}\;,
\end{align}
where $c_i = \cos(\alpha_i)$ and $s_i = \sin(\alpha_i)$, $\alpha_i \in [-\pi,\pi)$.
\end{lemma}

\section{Protocol template}

We introduce a general template that a test of quantumness may take. The template divides the test into two phases, \emph{Phase A} and \emph{Phase B}. Phase A is a ``setup phase'' in which the verifier and prover exchange classical information that, informally, guarantees that the prover has properly set up their workspace. The phase may include some tests, at the end of which the verifier may decide to abort the phase and either accept or reject the prover's actions outright. This possibility is captured by an outcome $\flag\in\{\acc,\rej,\cont\}$ that the verifier may return. Here, $\flag=\cont$ means that no decision has been taken and the protocol should proceed to Phase B. In Phase B, a single-bit challenge $m$ is issued by the verifier to the prover, who responds with a single-bit outcome $b$. The value $b$ returned by the prover is checked against a correct value $\hat{c}_m$ that is computed by the verifier as a function of its private randomness $\rand$, the transcript $\trans$ of the interaction in the first phase, and  the challenge bit $m$. This template is summarized in Figure~\ref{fig:template}.

\begin{figure}[!htbp]
  \centering
  \begin{gamespec}
Fix a security parameter $\lambda$. \\
  Phase A:
		\begin{enumerate}
      \setlength\itemsep{1pt}
    \item  The verifier and prover interact classically. At the end of the interaction, the verifier returns a $\flag\in\{\acc,\rej,\cont\}$.
		Let $\rand$ denote the verifier's random bits used in that phase, $\trans$ the transcript of the interaction, and $\ket{\psi_\trans}\in \mH_{\reg{PE}}$ the state of the prover and the environment at the end of the interaction.
		\end{enumerate}
		Phase B (executed only in case $\flag=\cont$):
		\begin{enumerate}
      \setlength\itemsep{1pt}
   \item The verifier sends a uniformly random challenge $m\in \{0,1\}$ to the prover.
	\item The prover returns a bit $b\in \{0,1\}$ to the verifier.
	\item The verifier accepts if and only if $(-1)^b = \hat{c}_m$, where $\hat{c}_m=\hat{c}_m(\rand,\trans)\in \{-1,1\}$ is a value computed by the verifier.
    \end{enumerate}
  \end{gamespec}
  \caption{Our template for a test of quantumness.}
  \label{fig:template}
\end{figure}

Our main results bound the maximum success probability of classical or quantum provers in any such protocol, assuming that a specific prediction task associated with the protocol is hard (informally, predicting the parity $\hat{c}_0\cdot \hat{c}_1$). To formulate the results we need to model the behavior of an arbitrary prover in the protocol. Such a prover is specified by a Hilbert space $\mH_\rP$, an initial quantum state $\ket{\psi}\in \mH_{\reg{P}} \otimes \mH_{\reg{E}}$, where $\mH_{\reg{E}}$ models an ``environment'' to which the prover does not have access (nor the verifier), and two families of measurements corresponding to the prover's actions in the two phases of the protocol (Figure~\ref{fig:template}). We emphasize that while Phase B naturally consists of a single round of interaction, Phase A may consist of multiple rounds of interaction. In this case, the actions of the prover in Phase A are described by multiple families of measurements, which incorporate any unitaries that the prover may apply to update its quantum state from one round to the next. Since our analysis will for the most part focus on the prover's behavior in the second phase, we abstract some of the details in the following definition.

\begin{definition}[Quantum Device]\label{def:device}
A \emph{quantum device} $\device$ is specified by:
\begin{itemize}
    \item Hilbert spaces $\mH_\rP$ and $\mH_\rE$, and a family of states $\{\ket{\psi_\trans}_{\reg{PE}}\,: \; \trans\in \mT\}$ on $\mH_\rP\otimes \mH_\rE$ together with a distribution $\mu$ on $\mT$. Here $\mT$ is used to denote the space of possible transcripts.
    \item For each $m\in \{0,1\}$, a projective measurement $\{\Pi^m_{b}\}$ on $\mH_{\reg{P}}$.\footnote{By Naimark's theorem the requirement that the measurement is projective is without loss of generality, up to enlarging the prover's Hilbert space with a single auxiliary qubit.}
\end{itemize}
\end{definition}

A quantum device can be used to specify a quantum prover in the template protocol as follows. The prover starts the protocol in a state $\ket{\psi} \in \mH_{\reg{P}}\otimes \mH_{\reg{E}}$. In Phase A, the verifier and the prover interact. This interaction results in a transcript $\trans$, obtained with probability $\mu(\trans)$, and a post-interaction state $\ket{\psi_\trans}_{\reg{PE}}$. Note that the fact that $(\trans, \ket{\psi_\trans})$ must be produced through a valid execution of Phase A of the protocol is implicit in the definition of a device. In the second phase, after having received $m$ the prover measures $\mH_P$ using $\{\Pi^{m}_{b}\}$ and returns the obtained outcome $b$. This is without loss of generality, as the measurement operators may incorporate any $m$-dependent unitary that the prover applies to their state after having received the verifier's challenge bit.

\section{Soundness analysis}

In this section we show that general statements on the soundness of the protocol template can be derived from a simple assumption about the hardness of predicting the parity $\hat{c}_0\cdot\hat{c}_1$ of the correct answers on challenges $m=0$ and $m=1$. Specifically, we establish \emph{classical soundness} (a bound on the maximum probability of success of any classical prover in the protocol), \emph{quantum soundness} (the same, for quantum provers), and the property of being a \emph{test for a qubit} (informally, that any quantum prover that succeeds with near-optimal probability must do so by performing measurements in Phase B such that the measurement applied for $m=0$ and the one applied for $m=1$ are close to maximally anti-commuting).

\subsection{The parity adversary}

We start by showing that any device that succeeds in the protocol template with some probability can be turned into an algorithm for predicting the parity of the verifier's two decision bits $\hat{c}_0$ and $\hat{c}_1$ in Phase B of the protocol. Although rather simple, observing such a transformation in the general setup of the protocol template is arguably our main conceptual contribution. Specifically, we construct the following.

\begin{definition}
Let $\delta,\kappa:\N\to [0,1]$.
We say that a pair $\mA = (\mA_1,\mA_2)$ of (classical or quantum) polynomial-time algorithms is a \emph{(classical or quantum) parity adversary with advantage $(\kappa,\delta)$} on (some instantiation of) the template protocol if the following hold. Firstly, $\mA_1$ and $\mA_2$ have the following structure:
\begin{itemize}
\item $\mA_1$ is a family of algorithms for the prover in an interaction with the verifier for Phase A of the template protocol. In particular, $\mA_1$ is initialized in a quantum state $\ket{\psi}\in \mH_{\reg{PE}}$ and completes the interaction by returning a transcript $\trans$ and a post-interaction state $\ket{\psi_\trans}\in \mH_{\reg{PE}}$.
\item $\mA_2$ takes as input the output $(\trans,\ket{\psi_{\trans}})$ of $\mA_1$. It returns a bit $b\in \{0,1\}$. ($\mA_2$ does not interact with the verifier.)
\end{itemize}
Secondly, it holds that the interaction of $\mA_1$ with the verifier results in $\flag=\rej$ with probability at most $\kappa(\lambda)$, and furthermore
\[\Es{\trans\leftarrow\mA_1(1^\lambda)} \Big|\Pr_{b\leftarrow \mA_2(1^\lambda,\trans,\ket{\psi_{\trans}})}\big( (-1)^ b= \hat{c}_0 \cdot\hat{c}_1\big) -\frac{1}{2} \Big|\,\leq\,\delta(\lambda)\;, \]
where the outer expectation is taken over $\trans$ generated from $\mA_1$ and is taken conditioned on $\flag=\cont$, the inner probability is taken over $b$ generated from $\mA_2$ on input $\trans$ and $\ket{\psi_\trans}$, and for $m\in\{0,1\}$, $\hat{c}_m$ is computed from $\trans$ and the verifier's private coins $\rand$ used when interacting with $\mA_1$ as in Phase B of the template protocol.
\end{definition}

Let $\device$ be a device for the protocol template (Definition~\ref{def:device}). Let $\mA_1$ execute Phase A of the protocol by interacting the device with the verifier, resulting in a transcript $\trans$ obtained with probability $\mu(\trans)$ and a post-execution state $\ket{\psi_{\trans}} \in \mH_{\reg{PE}}$. Let $\mA_2$ perform the following actions. $\mA_2$ first applies the projective measurement $\{\Pi^0_b\}$ on $\mH_{\reg{P}}$ to obtain a $b_0\in \{0,1\}$. Then $\mA_2$ applies the projective measurement $\{\Pi^1_b\}$ on $\mH_{\reg{P}}$ to obtain a  $b_1\in \{0,1\}$. Finally, $\mA_2$ returns $b_0\oplus b_1$. This construction is summarized in Figure~\ref{fig:parity-adv}.

\begin{figure}[!htbp]
  \centering
  \begin{gamespec}
	Fix a device $\device$ for the protocol template.
    \begin{itemize}
      \setlength\itemsep{1pt}
    \item Algorithm $\mA_1$: Execute Phase A of the protocol template to obtain $(\trans,\ket{\psi_{\trans}})$ and a flag $\flag$ returned by the verifier.
    \item Algorithm $\mA_2$: If $\flag=\cont$,
		\begin{enumerate}
		      \setlength\itemsep{1pt}
		\item Measure $\mH_{\reg{P}}$ using the two-outcome measurement $\{\Pi^0_0,\Pi^0_1\}$. Let $b_0\in \{0,1\}$ be the outcome obtained.
    \item Measure using the two-outcome measurement $\{\Pi^1_0,\Pi^1_1\}$. Let $b_1\in \{0,1\}$ be the outcome obtained.
\item Return $b_0\oplus b_1$.
\end{enumerate}
\end{itemize}
  \end{gamespec}
  \caption{A parity adversary.}
  \label{fig:parity-adv}
\end{figure}

Given a choice of random bits $\rand$ and a transcript $\trans$, define projections on $\mH_{\reg{P}}$ by
\begin{equation}\label{eq:def-q-a}
Q_0 \,=\, \Pi^0_{\hat{c}_0}\qquad\text{and}\qquad Q_1 \,=\, \Pi^1_{\hat{c}_1}\;.
\end{equation}

\begin{lemma}\label{lem:xor-suc}
Conditioned on $\flag=\cont$ and $\trans$ having been obtained after the execution of $\mA_1$, the parity adversary defined in Figure~\ref{fig:parity-adv} returns an outcome $b\in\{0,1\}$ such that
\begin{align*}
 p_{\mathrm{xor}} &= p_{\mathrm{xor}}(\rand,\trans)\\
&:= \Pr\big( (-1)^b=\hat{c}_0\cdot \hat{c}_1\big)\\
&=\big\| ( Q_{1}Q_{0}+(\Id-Q_{1})(\Id-Q_{0}) )\ket{{\psi}_\trans}\big\|^2\;.
\end{align*}
\end{lemma}

\begin{proof}
It holds that $(-1)^{b_0\oplus b_1}=\hat{c}_0\cdot \hat{c}_1$ if and only if either ($(-1)^{b_0}=\hat{c}_0$ and $(-1)^{b_1}=\hat{c}_1$) or ($(-1)^{b_0}\neq\hat{c}_0$ and $(-1)^{b_1}\neq\hat{c}_1$). By definition, the probability of the first event is  $\|Q_1Q_0\ket{{\psi}_\trans}\|^2$ and the probability of the second event is $\|(\Id-Q_1)(\Id-Q_0)\ket{{\psi}_\trans}\|^2$. Using that $Q_1$ and $(\Id-Q_1)$ are orthogonal, the lemma follows.
\end{proof}

\subsection{Classical and quantum soundness}
\label{sec:soundness}

In this section we show that a precise bound on the classical and quantum soundness of the protocol template can be obtained from the following assumption.

\begin{assumption}\label{ass:parity}
There is a function $s:\N\times [0,1] \to [0,1]$ such that for any (classical or quantum) polynomial-time parity adversary $\mA$ with advantage $(\kappa,\delta)$, it holds that $\delta(\lambda) \leq s(\lambda,\kappa(\lambda))$.
\end{assumption}

Note that Assumption~\ref{ass:parity} will in general be conditional on the hardness of some computational problem, such as the Learning with Errors problem. In the examples from Section~\ref{sec:applications} we will give various instantiations of the assumption. As we will see, given a concrete protocol that fits the protocol template it is generally quite straightforward to show that the assumption holds (often, it will hold for a function $s(\lambda)=\negl(\lambda)$, and the possibility for executing a test in Phase A will not even be used). However, the conclusion that we obtain on quantum soundness of the protocol template will comparatively be quite strong.

Let $\device$ be a (classical or quantum) polynomial-time device for the protocol template. Recall the definition of the projections $Q_0$ and $Q_1$ in~\eqref{eq:def-q-a} (which implicitly depend on the verifier's private randomness $\rand$  and the transcript $\trans$ from Phase A).
By definition the probability that $\mathfrak{D}$ succeeds in Phase B of the protocol, conditioned on $\flag=\cont$ and $m\in\{0,1\}$, is
\begin{equation}\label{eq:def-pm}
p_{m}=\lVert Q_{m} \ket{{\psi}_\trans}\rVert^{2}\;.
\end{equation}
Thus the probability that the device succeeds in Phase B of the protocol is
\begin{align}
\frac{1}{2}(p_{0}+p_{1})=\frac{1}{2}\lVert Q_{0} \ket{{\psi}_\trans}\rVert^{2} + \frac{1}{2}\lVert Q_{1} \ket{{\psi}_\trans}\rVert^{2}\;.
\end{align}
By Jordan's lemma (Lemma~\ref{lem:jordan}) there is a decomposition $\mH_{\reg{PE}}=\oplus_i \mH_i$ such that for all $i$, $\dim(\mH_i)\leq 2$ and moreover $\mH_i$ is invariant under both $Q_0$ and $Q_1$. 
For $\gamma \in \{1,2\}$ let $\mS_\gamma$ be the collection of indices $i$ such that $\dim(\mH_i)=\gamma$ and let $\mS=\mS_1\cup\mS_2$.

Fix an index $i\in\mS$. 
Let $\ket{u_{i}}$ be the normalized projection of $\ket{{\psi}_\trans}$ on $\mH_i$, and let $t_{i}=|\bra{u_{i} }{\psi}_\trans\rangle|^{2}$. 
If $i\in \mS_1$ then the restrictions of $Q_0$ and $Q_1$ to $\mH_i$ take the form
\begin{equation*}
Q_0 \,=\, \begin{pmatrix} \cos^2 \alpha_i \end{pmatrix}\quad\text{ and }\quad Q_1\,=\,\begin{pmatrix} \cos^2 \beta_i \end{pmatrix}\;,
\end{equation*}
where $\alpha_i,\beta_i\in \{0,\frac{\pi}{2}\}$. 
If $i\in \mS_2$ then there exists a state in $\mH_{i}$ which is orthogonal to $\ket{u_{i}}$, denote it by $\ket{u^{\perp}_{i}}$.
The pair $\{\ket{u_{i}},\ket{u^{\perp}_{i}}\}$ is an orthonormal basis for $\mH_{i}$ in which $Q_0$ and $Q_1$ take the form
\begin{equation}\label{eq:def-q0}
Q_0\,=\,\begin{pmatrix}
\cos^{2}(\alpha_{i}) & \cos(\alpha_{i})\sin(\alpha_{i})\\
\cos(\alpha_{i})\sin(\alpha_{i}) & \sin^{2}(\alpha_{i})
\end{pmatrix}
\end{equation}
and
\begin{equation}\label{eq:def-q1}
 Q_1\,=\,\begin{pmatrix}
\cos^{2}(\beta_{i}) & \cos(\beta_{i})\sin(\beta_{i})\\
\cos(\beta_{i})\sin(\beta_{i}) & \sin^{2}(\beta_{i})
\end{pmatrix}\end{equation}
respectively, for some $\alpha_{i},\beta_{i}\in [-\pi/2,\pi/2)$.\footnote{Without loss of generality, both $Q_0$ and $Q_1$ have rank exactly $1$ in $\mH_i$. In all other cases, the $2$-dimensional space $\mH_i$ can be further decomposed as a sum of two invariant $1$-dimensional spaces.}
With these notations, starting from~\eqref{eq:def-pm} one easily verifies that
\begin{align}
    p_{0} &= \sum_{i }t_{i}\cos^{2}(\alpha_{i})\;,\notag\\
    p_{1} &=\sum_{i}t_{i}\cos^{2}(\beta_{i})\;.\label{eq:p0-ti}
\end{align}
We summarize our findings so far in the following claim.

\begin{nclaim}\label{claim:probs}
The probability that the device $\device$ succeeds in the protocol template, conditioned on $(\rand,\trans)$ and $\flag=\cont$ having been obtained in Phase A, is
\begin{align}
\frac{1}{2}\big(p_{0}+p_{1})&= \sum_{i}\frac{t_{i}}{2}(\cos^{2}(\alpha_{i})+\cos^{2}(\beta_{i}))\;.\footnotemark
\label{eq:p-succ}
\end{align}
Furthermore, the parity adversary derived from $\device$ as in Figure~\ref{fig:parity-adv} has advantage (also conditioned on $\rand$ and $\trans$)
\begin{align}
\delta &=\Big|p_\mathrm{xor} - \frac{1}{2}\Big|\;.
\label{eq:delta-pxor}
\end{align}
where
\begin{equation}
p_\mathrm{xor} = \sum_{i}t_{i}\cos^{2}(\alpha_{i}-\beta_{i})
\label{eq:p-xor}
\end{equation}

\end{nclaim}

\begin{proof}
The first part of the claim follows directly from~\eqref{eq:p0-ti}. The second part follows by direct calculation using the expression from Lemma~\ref{lem:xor-suc}. Specifically, starting from the expressions in~\eqref{eq:def-q0} and~\eqref{eq:def-q1} we obtain
\begin{align*}
Q_{1}Q_{0}+(\Id-Q_{1})(\Id-Q_{0})
&=
\begin{pmatrix}
c^2 \hat{c}^2 + s^2 \hat{s}^2 + 2c \hat{c}s\hat{s} &  \hat{c}^2 c s + s^2 \hat{c}\hat{s} - \hat{s}^2 c s - c^2 \hat{c}\hat{s}\\
c^2 \hat{c}\hat{s} + \hat{s}^2 cs - s^2 \hat{c}\hat{s} - \hat{c}^2 cs & s^2 \hat{s}^2 + c^2 \hat{c}^2 + cs\hat{c}\hat{s}\end{pmatrix}\\
&=\begin{pmatrix}
\cos^2(\alpha_i-\beta_i) & \cos(\alpha_i-\beta_i)\sin(\alpha_i-\beta_i) \\ -\cos(\alpha_i-\beta_i)\sin(\alpha_i-\beta_i) & \cos^2(\alpha_i-\beta_i)\end{pmatrix}
\;,
\end{align*}
where for the middle expression we used the shorthand $c=\cos(\alpha_i)$, $s=\sin(\alpha_i)$, $\hat{c}=\cos(\beta_i)$ and $\hat{s}=\sin(\beta_i)$ and for the last line we used the trigonometric identities
\begin{align*}
 \cos(\alpha_i-\beta_i) &=  \cos(\alpha_i)\cos(\beta_i)+\sin(\alpha_i)\sin(\beta_i)\;,\\
\sin(\alpha_i-\beta_i) &= \sin(\alpha_i)\cos(\beta_i)-\cos(\alpha_i)\sin(\beta_i)\;.
\end{align*}
This allows us to verify that $\|(Q_{1}Q_{0}+(\Id-Q_{1})(\Id-Q_{0}))\ket{u_i}\|^2 = \cos^2(\alpha_i-\beta_i)$, establishing the claim.
\end{proof}

\begin{theorem}[Classical and quantum soundness]\label{thm:main}
Suppose that Assumption~\ref{ass:parity} holds for some function $s(\lambda,\kappa)$. Then the maximum probability with which a classical (resp. quantum) polynomial-time prover which succeeds in Phase A of the protocol template with probability at least $1-\kappa$ may succeed in Phase B is $\frac{3}{4}+\frac{1}{2}\,s(\lambda,\kappa)$ (resp. $\cos^2 \frac{\pi}{8} + \,s(\lambda,\kappa)$).
\end{theorem}

\begin{proof}
For the proof, we fix a device $\device$ and use the notation introduced towards the proof of Claim~\ref{claim:probs}.

We first show classical soundness. In this case, $\mS_2 = \emptyset$ and thus by~\eqref{eq:p-xor},
\[ p_{\mathrm{xor}} = \sum_{i: \alpha_i = \beta_i } t_{i}\;.\]
It follows that, from~\eqref{eq:p-succ},
\begin{align*}
p_{0}+p_{1}&=\sum_{i}t_{i}(\cos^{2}(\alpha_{i})+\cos^{2}(\beta_{i}))\\
&\leq \sum_{i:\alpha_i=\beta_i} 2t_{i} + \sum_{i:\alpha_i\neq \beta_i} t_i = 2p_\mathrm{xor} + (1 - p_\mathrm{xor}) \\
&\leq \frac{3}{2}+\delta\;,
\end{align*}
Here the second line uses $\sum_i t_i=1$ and~\eqref{eq:p-xor} and the third uses~\eqref{eq:delta-pxor}. We then have that the probability with which the classical prover succeeds in Phase B is $\frac{1}{2}(p_0 + p_1) \leq \frac{3}{4} + \frac{1}{2}\,s(\lambda, \kappa),$ since by Assumption~\ref{ass:parity}, $\delta \leq s(\lambda, \kappa).$

Next we show quantum soundness. 
We use the following inequality (derived in Appendix~\ref{appendix:trig}), valid for any $\alpha,\beta$: 
\begin{equation}\label{eq:num-ineq}
\cos^2(\alpha) + \cos^2(\beta) \leq \left| 2\cos^2(\alpha - \beta) - 1 \right| + 2 \cos^2(\pi/8)
\end{equation}
It follows that, starting from~\eqref{eq:p-succ},
\begin{align*}
\frac{1}{2} (p_{0}+p_{1}) &=\frac{1}{2}\sum_{i}t_{i}(\cos^{2}(\alpha_{i})+\cos^{2}(\beta_{i}))\\
&\leq \sum_i t_i \left| \cos^2(\alpha_i-\beta_i)-\frac{1}{2} \right| + \cos^2 \frac{\pi}{8}\\
&\leq \delta + \cos^2\frac{\pi}{8}\;,
\end{align*}
where the second line follows from~\eqref{eq:num-ineq} together with the fact that $\sum_i t_i = 1$ and the third follows from~\eqref{eq:p-xor}.
We therefore have that $\frac{1}{2}(p_0 + p_1) \leq s(\lambda,\kappa) + \cos^2 \frac{\pi}{8}$ as claimed.
\end{proof}

\subsection{Qubit test}
\label{sec:qubit}

In this section we go beyond quantum soundness and show that the protocol template can be used to certify that any prover which succeeds with probability close to the quantum optimum of $\omega=\cos^2 \frac{\pi}{8}$ ``has a qubit.'' The key technical step is given by the following proposition.

\begin{proposition}\label{prop:qubit}
Let $\device$ be a polynomial-time quantum device that is such that $\Pr(\flag=\rej)\leq \kappa$ in Phase A of the protocol template, and that succeeds with probability $\omega- \eps$ in Phase B (conditioned on $\flag=\cont$ having been returned in Phase A). Let $(t_i)$, $(\alpha_i)$, $(\beta_i)$ be defined as in the start of Section~\ref{sec:soundness}. Then there is an $\eta = O(\eps+\sqrt{s})$ such that for all but a fraction at most $\eta$ of transcripts $\trans$ (such that $\flag=\cont)$ and indices $i$, as measured by $(t_i)$, it holds that
\[ \big| \alpha_i  \pm \frac{\pi}{8}\big| \,\leq\, \eta\qquad\text{and}\qquad \big|\beta_i+\alpha_i\big| \,\leq\,\eta\;,\]
for some choice of sign $\pm$ (depending on $i$).
\end{proposition}

\begin{proof}
Recall the set $\mS$, and let
\[\mS' = \{i: \alpha_i \notin [-3\pi/16,3\pi/16]\} \cup \{i: \beta_i \notin [-3\pi/16,3\pi/16]\}\;.\]
We start by bounding $\sum_{i\in \mS'} t_i$. For this, we first observe the following inequality. For $\alpha \in [-\pi/2,-3\pi/16]\cup[3\pi/16,\pi/2]$ it holds that
\begin{equation}\label{eq:num-ine2}
\cos^2(\alpha-\beta) - \frac{1}{2} \,\geq\, 100\Big(\frac{\cos^2\alpha+\cos^2\beta}{2} -0.851\Big)\;.
\end{equation}
This is because for the range of $\alpha$ indicated, the right-hand side is always less than $-0.5$ (while the left-hand side is always at least $-0.5$), as can be verified by direct calculation. Using both~\eqref{eq:num-ineq} and~\eqref{eq:num-ine2} it follows from~\eqref{eq:p-succ} that
\begin{align*}
100(\omega - \eps) &\leq \sum_i t_i \Big(\cos^2(\alpha_i-\beta_i)-\frac{1}{2}\Big) + 85.1\sum_{i\in\mS'} t_i + 100\omega \sum_{i\notin \mS'} t_i\;.
\end{align*}
Using Assumption~\ref{ass:parity} and $\sum_i t_i=1$ we get that
\begin{equation}\label{eq:sprime-bound}
\sum_{i\in \mS'} t_i \,\leq\, \frac{ 1}{\omega-0.851}\,\eps + O(s(\lambda,\kappa)) \,=\,O(\eps+s)\;.
\end{equation}

Since the function $x\mapsto \cos^2(x)$ is strictly concave on $[-3\pi/16,3\pi/16]$, for $\alpha,\beta\in [-3\pi/16,3\pi/16]$ it holds that
\begin{equation}
\zeta \cos^2(\alpha) + (1 - \zeta) \cos^2(\beta)\, =\, \zeta \cos^2(\alpha) + (1 - \zeta) \cos^2(-\beta) \leq \cos^2(\zeta \alpha - (1 - \zeta)\beta)\;,
\end{equation}
for all $\zeta \in [0,1]$. Taking $\zeta = \frac{1}{2}$, we have
\begin{equation}\label{eq:conv-a}
\frac{1}{2}\Big(\cos^2\alpha + \cos^2 \beta\Big)\, \leq\, \cos^2\Big(\frac{\alpha-\beta}{2}\Big)\;.
\end{equation}
It follows that
\begin{equation}\label{eq:conv}
\frac{1}{2}\sum_{i\notin \mS'} t_i (\cos^{2}(\alpha_{i})+\cos^{2}(\beta_{i})) \,\leq\,\sum_{i\notin \mS'} t_i \cos^2\Big(\frac{\alpha_i-\beta_i}{2}\Big)\;.
\end{equation}
Starting from expression~\eqref{eq:p-succ} we then deduce that,
\begin{align}
\omega - \eps &= \frac{1}{2}\sum_i t_i \big( \cos^2(\alpha_i) + \cos^2(\beta_i) \big) \notag\\
&\leq   \sum_{i\notin \mS'} t_i \cos^2 \Big( \frac{\alpha_i - \beta_i}{2}\Big)+O(\eps+s)\notag\\
&\leq  \sum_{i} \frac{t_i}{2} \big( 1 + \big|\cos\big( \alpha_i - \beta_i\big)\big|\big)+O(\eps+s) \notag\\
&\leq \frac{1}{2}  + \frac{1}{2}\sqrt{\sum_{i} t_i \cos^2 \big(\alpha_i-\beta_i)}+ O(\eps+s)\;,\label{eq:conv-ineq}
\end{align}
where the second line uses~\eqref{eq:sprime-bound} and~\eqref{eq:conv}, the third line adds non-negative terms for $i\in \mS'$ and uses the trigonometric identity $\cos^2(x/2)=\frac{1}{2}(1+\cos(x))$, and the last is by concavity of the square root function. By Assumption~\ref{ass:parity} the right-hand side is at most
\[\frac{1}{2}+\frac{1}{2\sqrt{2}} + O(\eps+s) + O\big(\sqrt{s}\big)\,=\, \omega + O\big(\eps+\sqrt{s}\big)\;.\]
Hence all inequalities in the derivation of~\eqref{eq:conv-ineq} must be tight up to $O(\eps+\sqrt{s})$. We show that this implies the following bounds.

\begin{nclaim}\label{claim:temp} The following inequalities hold:
\begin{align}
\sum_i t_i \Big( |\alpha_i - \beta_i|-\frac{\pi}{4}\Big)^2 &= O\big(\eps+\sqrt{s}\big)\;. \label{eq:tight-1}\\
\sum_i t_i \big(\alpha_i + \beta_i\big)^2 &= O\big(\eps+\sqrt{s}\big)\;.\label{eq:tight-2}
\end{align}
\end{nclaim}

\begin{proof}
We first prove~\eqref{eq:tight-1}. For this we exploit near-tightness of the application of Jensen's inequality on the last line of~\eqref{eq:conv-ineq}. This immediately implies that
\begin{equation*}
\sum_{i} t_i \Big( \big| \cos(\alpha_i-\beta_i)\big| - \sqrt{\sum_{i'} t_{i'} \cos^2 \big(\alpha_{i'}-\beta_{i'})}\Big)^2 \,=\, O\big( \eps+\sqrt{s}\big)\;.
\end{equation*}
By definition, $\Big| \sum_{i'} t_{i'} \cos^2 \big(\alpha_{i'}-\beta_{i'}) - \frac{1}{2} \Big| \leq s$, hence
\begin{equation*}
\sum_{i} t_i \Big( \big| \cos(\alpha_i-\beta_i)\big| - \frac{1}{\sqrt{2}}\Big)^2 \,=\, O\big( \eps+\sqrt{s}\big)\;.
\end{equation*}
Using that for $x\in [0,\pi]$,
\[ \Big| \cos(x)-\frac{1}{\sqrt{2}} \Big| \geq \frac{1}{2}\Big|x-\frac{\pi}{4}\Big|\;,\]
Eq.~\eqref{eq:tight-1} follows. To show~\eqref{eq:tight-2}, we similarly use near-tightness in the second line of~\eqref{eq:conv-ineq}, using the strict concavity expressed in~\eqref{eq:conv-a}.
\end{proof}

Applying Markov's inequality to the conclusions of Claim~\ref{claim:temp}, the proposition follows.
\end{proof}

To formulate the qubit test we introduce the observables
\begin{equation}\label{eq:def-s}
 S_0 \,=\, 2Q_0-\Id\quad\text{and } S_1 = 2Q_1-\Id\;.
\end{equation}
The next theorem states that the observables $S_0$ and $S_1$ must be close to anti-commuting, as measured by the squared norm of the anti-commutator when evaluated on the state $\ket{{\psi_\trans}}$.

\begin{theorem}\label{thm:qubit}
Suppose that Assumption~\ref{ass:parity} holds for some function $s(\lambda,\kappa)$.
Let $\device$ be a polynomial-time quantum device that  is such that $\Pr(\flag=\rej)\leq \kappa$ in Phase A of the protocol template, and that succeeds with probability $\omega- \eps$ in Phase B (conditioned on $\flag=\cont$ having been returned in Phase A). Then, on average over $(\rand,\trans)$ and $\ket{\psi_{\trans}}$ generated in Phase A of the protocol and conditioned on $\flag=\cont$, it holds that
\begin{equation}\label{eq:ac-bound}
 \bra{{\psi_\trans}} \{ S_0,S_1\}^2 \ket{{\psi_{\trans}}} \,=\, O\big(\eps+\sqrt{s}\big)\;,
\end{equation}
where $\{S_0,S_1\}=S_0S_1+S_1S_0$.
\end{theorem}

Before giving the proof of the theorem, which follows from Proposition~\ref{prop:qubit} by direct calculation, we motivate it by discussing its implications. A first, rather immediate consequence of a bound such as~\eqref{eq:ac-bound} is that there exists an isometry $V:\mH_\reg{PE} \to \C^2 \otimes \mH'_{\reg{E}}$ such that, under the isometry, $S_0 \simeq_{O(\sqrt{\eps})} \sigma_Z \otimes \Id_{\mH_{\reg{E}}'}$ and $S_1 \simeq_{O(\sqrt{\eps})} \sigma_X\otimes \Id_{\mH_{\reg{E}}'}$.
Here, $\simeq$ measures distance in the appropriate state-dependent norm, and $\sigma_Z$ and $\sigma_X$ are the canonical Pauli matrices. This statement is standard in the self-testing literature. (For a proof and more details, see for example~\cite[Lemma 2.9]{gheorghiu2019computationally}. A more extensive discussion appears in~\cite[Lecture 2]{vidick2020course}.) This statement formalizes the intuition that any successful prover in the qubit test must ``have a qubit:'' the operations that it performs in Phase B of the protocol template are essentially equivalent, up to isometry, with measurements in the standard ($\sigma_Z)$ or Hadamard ($\sigma_X$) basis on a qubit.

This statement, of being a ``test of a qubit,'' is powerful. In~\cite{brakerski2021cryptographic} a similar statement is used to obtain certifiable randomness. In~\cite{AM23} it is shown that the specific bound shown in Theorem~\ref{thm:qubit} suffices to obtain precise quantitative bounds on the amount of randomness generated in an execution of our template protocol, as well as on the accumulation of randomness through multiple sequential executions. As a consequence of our work, their results also imply that certified randomness accumulation can be achieved using any of the concrete instantiations given in Section~\ref{sec:applications}. 
In~\cite{metger2021self} this is expanded in a test for an EPR pair, which can lead to protocols for device-independent quantum key distribution as in~\cite{metger2021device}. In~\cite{gheorghiu2019computationally} the qubit test forms the basis for a protocol for classical delegation of quantum computation. While we do not work out any of these applications, our results open the door to developing them based on any protocol that follows our template (such as the examples given in Section~\ref{sec:applications}) by using known techniques.

\begin{proof}[Proof of Theorem~\ref{thm:qubit}]
We use the notation introduced at the start of Section~\ref{sec:soundness}. Fix an index $i \in \mS_2$. Then
\begin{equation}
S_0 \,=\, \begin{pmatrix} \cos(2\alpha_i) & \sin(2\alpha_i) \\ \sin(2\alpha_i) & -\cos(2\alpha_i) \end{pmatrix}\quad\text{and}\quad S_1 \,=\, \begin{pmatrix} \cos(2\beta_i) & \sin(2\beta_i) \\ \sin(2\beta_i) & -\cos(2\beta_i) \end{pmatrix}\;.
\end{equation}
This allows us to compute
\[ S_0S_1 + S_1 S_0 \,=\, 2 \cos(2(\alpha_i-\beta_i)) \Id\;.\]
In particular,
\[ \bra{u_i} \{S_0,S_1\}^2 \ket{u_i} \leq 4 \cos^2 \big(2(\alpha_i-\beta_i)\big)\;.\]
Using $\cos^2(x)\leq (x-\pi/2)^2$ for $x\in[\pi/4,3\pi/4]$ it follows that whenever $2(\alpha_i-\beta_i)\in[-\pi/4,3\pi/4]$,
\[ \bra{u_i} \{S_0,S_1\}^2 \ket{u_i} \leq 16\Big( \big|\alpha_i-\beta_i\big|-\frac{\pi}{4}\Big)^2\;.\]
Using Proposition~\ref{prop:qubit} all but a fraction $O(\eps+\sqrt{s})$ of indices $i$ satisfy this condition, and furthermore for these $i$ the right-hand side is $O(\eps+\sqrt{s})$. The theorem follows.
\end{proof}

\section{Applications}
\label{sec:applications}

We give three applications. First we consider the protocol from~\cite{KCVY2022}. This protocol is based on trapdoor claw-free functions, for which we recall the definition in the next section. Next we introduce a slightly simplified version of that protocol, and show that its proof of security follows in a completely direct way from our methods. Third we consider the general compiler from~\cite{kalai2022quantum} and apply it to the CHSH game.

\subsection{Trapdoor claw-free functions}

The main cryptographic primitive upon which the concrete protocols we describe rely on is the \emph{trapdoor claw-free function family} (TCF), which is defined as follows.

\begin{definition}\label{def:tcf}
Let $\lambda$ be a security parameter, $\mK$ a set of keys, and $\mX_{k}$ and $\mY_{k}$ finite sets for each $k\in \mK$. A family of functions
\begin{align*}\mathcal{F}=\{f_{k}:\mX_{k}\to \mY_{k}\}_{k\in \mK}\end{align*}
is called a trapdoor claw-free (TCF) function family if the following conditions hold:
\begin{enumerate}
    \item \textbf{Efficient Function Generation.} There exists an efficient probabilistic algorithm $\Gen$ which given a security parameter $\lambda$ in unary generates a key $k\in \mK$ and the associated trapdoor $t_{k}$:
    \begin{align*}(k,t_{k})\leftarrow \Gen(1^{\lambda})\end{align*}
    \item \textbf{Trapdoor Injective Pair.} For all keys $k\in \mK$, the following conditions hold.
    \begin{enumerate}
        \item Injective pair: There exists a perfect matching $R_{k}$ on $\mX_k$ such that for all $(x_{0}, x_{1})\in R_k$, $f_{k}(x_{0})=f_{k}(x_{1})$. 
        \item Trapdoor: There exists an efficient deterministic algorithm $\Inv_{k}$ such that for all $y\in \mY_{k}$ and $(x_{0},x_{1})$ such that $f_{k}(x_{0})=f_{k}(x_{1})=y$, $\Inv(t_{k},y)=(x_{0},x_{1})$. 
    \end{enumerate}
    \item \textbf{Claw-free.} For any non-uniform probabilistic polynomial time Turing machine $\mathcal{A}$,
    \begin{align*}\Pr\big( f_{k}(x_{0})=f_{k}(x_{1})\wedge x_{0}\neq x_{1} | (x_{0},x_{1})\leftarrow \mathcal{A}(k)\big)\,=\,\negl(\lambda)\;,\end{align*}  where the probability is over both the choice of $k$ and the random coins of $\mathcal{A}$.
    \item \textbf{Efficient Superposition.} There exists a polynomial-time quantum algorithm that on input a key $k$ prepares the state
    \begin{equation}\label{eq:claw-state}
		\frac{1}{\sqrt{|\mX_{k}|}}\sum_{x\in \mX_{k}}{\ket{x} \ket{f_{k}(x)}}\;.\end{equation}
\end{enumerate}
\end{definition}

Note that the third condition, claw-freeness, may be required to hold with regard to quantum adversaries, or only with regard to classical adversaries, depending on the application.

Trapdoor claw-free functions can be constructed based on a diversity of concrete assumptions. 
In~\cite{KCVY2022} two constructions are given, based on Rabin's function and based on the Decisional Diffie-Hellman problem. 
Neither assumption is secure against quantum adversaries (only classical ones). 
In~\cite{brakerski2021cryptographic} a variant called ``noisy'' trapdoor claw-free function is constructed based on the Learning with Errors (LWE) problem. 
Furthermore, in~\cite{brakerski2020simpler} this was extended to Ring-LWE, which is expected to be more efficient than standard LWE. 
It is straightforward to verify that the noisy type of TCF can also be used in our protocol; for the sake of clarity we describe the protocols using simpler ``non-noisy'' TCFs.

\subsection{The $3$-round protocol from~\cite{KCVY2022}}

In Figure~\ref{fig:kcvy} we recall the $3$-round ($6$-message) test of quantumness from~\cite{KCVY2022}, which we refer to as the \emph{KCVY protocol}. The protocol depends on a TCF family, of which we recall the definition in Definition~\ref{def:tcf}.

\begin{figure}[!htbp]
  \centering
  \begin{gamespec}
Fix a TCF family $(\Gen,\Inv)$ and a security parameter $\lambda$.
    \begin{enumerate}
      \setlength\itemsep{1pt}
    \item  The verifier samples $(k,t_{k})\leftarrow \Gen(1^{\lambda})$ and sends $k$ to the prover.
     The prover returns a string $y\in\mY_k$ to the verifier.
\item The verifier computes $(x_0,x_1)\leftarrow\Inv(t_k,y)$. The verifier decides to perform either of the following with probability $\frac{1}{2}$ each:
\begin{enumerate}
\item (Preimage test:) The verifier requests a preimage. The prover responds with a string $x$. The verifier accepts if and only if $x\in\{x_0,x_1\}$.
\item (Equation test:)
\begin{enumerate}
    \item The verifier chooses $r\in\{0,1\}^n$ uniformly at random and sends $r$ to the prover.
		 The prover responds with $d\in \{0,1\}^n$.
		\item The verifier sends a uniformly random $m\in \{0,1\}$ to the prover.
    The prover responds with  a bit  $b\in \{0,1\}$.
    \item The verifier accepts if and only if $(-1)^b=\hat{c}_m$, where $\hat{c}_m$ is defined in Table~\ref{probs1}.
    \end{enumerate}
		    \end{enumerate}
    \end{enumerate}
  \end{gamespec}
  \caption{The KCVY protocol.}
  \label{fig:kcvy}
\end{figure}

\begin{table}[ht]
  \centering
  \begin{tabular}{|l|r|r|r|r|}
    \hline
    $\ket{\phi_{r,d}}$ &
    $r\cdot (x_{0}\oplus x_{1})$ &  $d\cdot (x_{0}\oplus x_{1})$ &
    $\hat{c}_0 \;(m=0)$ & $\hat{c}_1\; (m=1)$\\
    \hline
    $\ket{0}$  & 0 & 0 & +1 & +1 \\
    $\ket{1}$  & 0 & 0 & -1 & -1\\
    $\ket{+}$  & 1 & 0 & +1 & -1\\
    $\ket{-}$  & 1 & 1 & -1 & +1\\
    \hline
  \end{tabular}
\caption{Here, $\ket{\phi_{r, d}}$ denotes the prover's state after step (b)i. in the protocol. The $\hat{c}$ column describes the $c$ that will likely be sent from an honest error-free prover in the equation test of the KCVY protocol, in case where $m=0$ or $m=1$.}
\label{probs1}
\end{table}

In~\cite{KCVY2022} it is shown that there exists an honest quantum prover that succeeds with probability $1$ in the preimage test, and with probability $\omega = \cos^2 \frac{\pi}{8}$ in the equation test. We complement their result by showing that any quantum polynomial-time prover that succeeds with probability at least $1-\kappa$ in the preimage test can succeed in the equation test with probability at most $\omega + O(\sqrt{\kappa})+\negl(\lambda)$. We do this by applying our main result, Theorem~\ref{thm:main}, to the KCVY protocol. The main observation needed is that for a given $y,d$ the product $\hat{c}_0\cdot \hat{c}_1 = (-1)^{r\cdot (x_0\oplus x_1)}$. By the claw-freeness property, this quantity should be hard to predict for a uniformly random $r$ --- as long as there is also a means of recovering $x_0$ or $x_1$, which is guaranteed by the preimage test. This allows us to establish Assumption~\ref{ass:parity} for this protocol and therefore use Theorem~\ref{thm:main}.

\begin{theorem}[Quantum soundness of the KCVY protocol]\label{thm:kcvy}
Suppose that a quantum polynomial-time prover succeeds in the preimage test with probability $p=1-\kappa$, and in the equation test with probability $q$. Then 
\[ q \,\leq\, \cos^2 \frac{\pi}{8} + \frac{\sqrt{\kappa}}{2} + \negl(\lambda)\;.\]
\end{theorem}

We remark that in principle using the same proof strategy as for the theorem, we could show that the KCVY protocol leads to a qubit test. We omit the details here and show this property for the simplified variant of the protocol introduced in the next section.

\begin{proof}
We first observe that the KCVY protocol fits the protocol template from Figure~\ref{fig:template} with the following adaptations. We incorporate all steps of the protocol except (b)ii. (b)iii. in Phase A. In particular, the choice of executing a preimage test or an equation test is made in Phase A. If the preimage test is chosen, then Phase A terminates with the result of that test, $\flag=\acc$ or $\flag=\rej$. If the equation test is chosen, then step (b)i. is executed in Phase A, $\flag=\cont$, and steps (b)ii. (b)iii. are executed in Phase B.

Now we need to show that Assumption~\ref{ass:parity} is satisfied. Suppose that $\mA_1$ succeeds with probability $1-\kappa$ in Phase A, and that $\mA_2$ has advantage $\delta$ in the parity guessing task. We first use $(\mA_1,\mA_2)$ to construct a quantum algorithm $\mA$' with the following properties. When given as input $(y,\ket{\psi_y})$ generated from the first round of the KCVY protocol (Figure~\ref{fig:kcvy}),
\begin{enumerate}
\item On input $m'=0$, $\mA'$ returns $x_0$ or $x_1$ with probability at least $1-\kappa$.
\item On input $m'=1$, $\mA'$ returns $x_0\oplus x_1$ with probability at least $4\delta^2$.
\end{enumerate}
To get the input for $\mA'$, we first execute $\mA_1$ for the first round of the protocol only. 
This yields a string $y$ and a post-measurement state $\ket{\psi_y}$. 
Now, if $m'=0$ then $\mA'$ executes the remaining actions of $\mA_1$ corresponding to the preimage test.
By supposition, the first item is satisfied.
If $m'=1$ then $\mA'$ proceeds as follows. 
$\mA'$ creates the state
\[\ket{\tilde{\psi}_y} \,=\, \frac{1}{\sqrt{2^n}} \sum_{r\in \{0,1\}^n} \ket{r} \otimes \ket{\psi_y}\;.\]
Let $U$ be the following unitary. $U$ first coherently executes the remainder of $\mA_1$ on this state, treating the first register as the verifier's question $r$ in step (b)i, and writes the outcome $d\in\{0,1\}^n$ in an ancilla register.
Then, $U$ coherently executes $\mA_2$ on all registers, writing the outcome $b$ in another ancilla register.
Observe that this unitary satisfies the two conditions of Definition~\ref{def:gl}, for $\eps = \delta$ and the string $a=x_0\oplus x_1$. This is because $\hat{c}_0\cdot \hat{c}_1 = r\cdot (x_0\oplus x_1)$, as can be verified from Table~\ref{probs1}.
Finally, $\mA'$ on input $m'=1$ executes the algorithm of Theorem~\ref{thm:qgl} (the quantum Goldreich-Levin theorem).
By Theorem~\ref{thm:qgl}, the second item above is satisfied.

To conclude we show the following.

\begin{nclaim}\label{claim:temp-1}
For any quantum polynomial time $\mA'$ satisfying the two items above, it holds that
\[ \delta \,\leq\, \frac{1}{2}\, \kappa^{1/2} + \negl(\lambda)\;\]
\end{nclaim}

\begin{proof}
We construct an algorithm $\mA''$ that returns both $x_0$ and $x_1$, thus violating the claw-free property of the TCF.
$\mA''$ is very simple: as above, it first executes $\mA_1$ for the first round of the protocol, yielding a string $y$ and state $\ket{\psi_y}$.
Then, it simply executes $\mA'$ on this state with input $m'=1$, and then immediately executes $\mA'$ again on the resulting state with input $m'=0$.
By item 2, with probability $4\delta^2$ the first execution of $\mA'$ obtains $x_0\oplus x_1$.
Let $P$ is the projection on this outcome being obtained (i.e.\ we model $\mA'$ as a projective measurement and $P$ is the projection associated with the outcome $x_0\oplus x_1$).
We have that $\bra{\psi_y} P \ket{\psi_y} \geq 4\delta^2$.
This then implies that $|\bra{\psi_y} P\ket{\psi_y}|^2 \geq 4 \delta^2 \|P\ket{\psi_y}\|^2$, or in other words, the trace distance between the original state, and the post-measurement state conditioned on the outcome being $x_0\oplus x_1$, is at most $\sqrt{1-4\delta^2}$. 
Therefore, when subsequently executing $\mA'$ on input $m'=0$ the outcome is $x_0$ or $x_1$ with probability at least $1-\kappa-\sqrt{1-4\delta^2}$. As long as this quantity is non-negligible, the claw-free property is violated. The claim follows.
\end{proof}

The theorem follows from Claim~\ref{claim:temp-1} by applying Theorem~\ref{thm:main}.
\end{proof}

\subsection{A simplified protocol}

We introduce a simplified variant of the KCVY protocol described in the previous section. 
This variant is described in Figure~\ref{fig:protocol}. 
Our variant introduces a small innovation that allows us to do away with the preimage test entirely. 
The idea is that, instead of sending a single string $r \in \{0,1\}^n$ with which the prover computes the inner products $r\cdot x_0$ and $r \cdot x_1$, the verifier sends separate strings $r_0, r_1 \in \{0,1\}^n$ and computes the inner products $r_0 \cdot x_0$ and $r_1 \cdot x_1$.\footnote{This does require a minor extra property of the TCF, which states that ``$x_0$'' type preimages can be efficiently distinguished from ``$x_1$'' type preimages; this property can be shown to hold for all TCF constructions of which we are aware.}
Then, predicting the parity $(r_0 \cdot x_0)\oplus (r_1 \cdot x_1)$ is equivalent to computing the value of $r' \cdot (x_0 || x_1)$ for a random string $r' = r_0 || r_1$, which by the quantum Goldreich-Levin theorem is as hard as predicting the string $x_0 || x_1$ and thus also the claw $(x_0,x_1)$. 

\begin{figure}[!htbp]
  \centering
  \begin{gamespec}
Fix a TCF family $(\Gen,\Inv)$ and a security parameter $\lambda$.
    \begin{enumerate}
      \setlength\itemsep{1pt}
    \item  The verifier samples $(k,t_{k})\leftarrow \Gen(1^{\lambda})$ and sends $k$ to the prover. The prover returns a string $y\in\mY_k$ to the verifier.
    \item The verifier chooses $r_{0},r_{1} \leftarrow \{0,1\}^n$ uniformly at random and sends them to the prover. The prover returns a string $d\in \{0,1\}^{n}$.
	\item The verifier sends a challenge $m\in \{0,1\}$ chosen uniformly at random to the prover. The prover responds with a bit  $b\in \{0,1\}$.
    \item The verifier computes $(x_0,x_1)\leftarrow\Inv(t_k,y)$. They accept if and only if $(-1)^b=\hat{c}_m$, where $\hat{c}_m$ is defined in~\eqref{def:hat-c}.
    \end{enumerate}
  \end{gamespec}
  \caption{A simpler variant of the KCVY protocol.}
  \label{fig:protocol}
\end{figure}

\subsubsection{Honest prover}

Because the protocol in Figure~\ref{fig:protocol} is new, we start by arguing quantum completeness: we describe the actions of a honest quantum prover in the protocol. 

\begin{proposition}
There exists a polynomial-time quantum prover who succeeds in the simplified KCVY protocol from Figure~\ref{fig:protocol} with probability $\cos^2(\pi/8)$.
\end{proposition}

\begin{proof}
Fix a TCF family $(\Gen,\Inv)$, a security parameter $\lambda$ and a key $(k,t_k)\leftarrow \Gen(1^\lambda)$ as generated by the verifier in the protocol. Furthermore, assume that $\mX_k \subseteq \{0,1\}^n$ for some $n=n(\lambda)$.

The prover proceeds as in the original protocol for the first round, yielding a string $y\in\mY_k$ and the post-measurement state
\[  \frac{1}{\sqrt{2}}\big( \ket{x_{0}}+\ket{x_{1}}\big)\;.\]
Upon receiving $r_0, r_1 \in\{0,1\}^{n}$, the prover computes an ancilla qubit which differentiates $x_0$ and $x_1$, yielding
\[  \frac{1}{\sqrt{2}}\big( \ket{0}\ket{x_{0}}+\ket{1}\ket{x_{1}}\big)\;.\]
With the ancilla, the prover can use controlled operations to compute the state
\[  \frac{1}{\sqrt{2}}\big( \ket{0}\ket{x_{0}}\ket{r_0 \cdot x_{0}}+\ket{1}\ket{x_{1}} \ket{r_1 \cdot x_1}\big)\;.\]
The prover then uncomputes the ancilla.
After a Hadamard transformation on the register containing $x_0$ and $x_1$, the state becomes
\[ \frac{1}{\sqrt{2^{2n+1}}}\sum_{d \in \{0,1\}^{n}} \ket{d} \left( (-1)^{d\cdot x_{0}} \ket{r_{0}\cdot x_{0}}+(-1)^{d \cdot x_{1}} \ket{r_{1}\cdot x_{1}}\right)\;.\]
 Measuring the first register to obtain a particular string $d$, the post-measurement state is
\[ \frac{(-1)^{d\cdot x_0}}{\sqrt{2}}\Big( \ket{r_{0}\cdot x_{0}}+(-1)^{d\cdot (x_{0} \oplus x_{1})}\ket{r_{1}\cdot x_{1}}\Big)\;.\]
The prover now returns the string $d$ to the verifier. They receive a challenge $m\in\{0,1\}$.
Let
\[\alpha = r_{0}\cdot x_{0}\oplus r_{1}\cdot x_{1} = r' \cdot (x_{0}||x_{1})\;,\qquad \beta = d \cdot \left( x_{0} \oplus x_{1} \right) \;,\]
and
\begin{align}\label{def:hat-c}
   \hat{c}_m \,=\, \hat{c}_m(r,x_{0},x_{1},d)\,=\,(1-\alpha)\cdot (-1)^{r_{0}\cdot x_{0}} + \alpha (-1)^{\beta} \cdot (-1)^{m}\;.
\end{align}
Finally, the prover measures the remaining qubit in the basis $\{\ket{\pi/8},\ket{5\pi/8}\}$ if $m=0$, and in the basis $\{\ket{-\pi/8},\ket{3\pi/8}\}$ in case $m=1$, where $\ket{\theta} = \cos(\theta)\ket{0} + \sin(\theta)\ket{1}$. Let $b\in \{0,1\}$ be the outcome obtained. It is straightforward to see that in all cases they return the correct answer $(-1)^b=\hat{c}_m$ with probability $\cos^2(\pi/8)$.
\end{proof}

\subsubsection{Soundness}

We now argue both classical and quantum soundness of the protocol, by applying Theorem~\ref{thm:main}. We obtain the following.

\begin{theorem}[Classical and quantum soundness of the simplified protocol]\label{thm:simplified}
The maximum probability with which a classical (resp. quantum) polynomial-time prover may succeed in the $3$-round protocol from Figure~\ref{fig:protocol} is $\frac{3}{4}+\negl(\lambda)$ (resp. $\cos^2 \frac{\pi}{8} + \negl(\lambda)$).
\end{theorem}

\begin{proof}
We emphasize that the proof is particularly simple.
To show the theorem, it suffices to (a) show that the protocol can be formatted as an instance of our protocol template, and (b) show that Assumption~\ref{ass:parity} holds for some function $s(\lambda,\kappa)$.

Step (a) is very direct: we simply combine the first two rounds of the protocol into Phase A, and the last round into Phase B. Step (b) is a bit more interesting, yet still straightforward. The main observation is that for a given $y,d$ and $r$, the product $\hat{c}_0\cdot\hat{c}_1= (-1)^{r\cdot(x_0\|x_1)}$. Therefore, an adversary able to predict the product can --- via the quantum Goldreich Levin algorithm --- recover a claw $(x_0,x_1)$.

In more detail, let $\mA=(\mA_1,\mA_2)$ be a parity adversary. Since there is no test in Phase A of the protocol, we can assume that $\kappa=0$. Suppose that $\mA_2$ has advantage $\delta$ in the parity guessing task. We use $(\mA_1,\mA_2)$ to construct a quantum algorithm $\mA$' that returns a claw $(x_0,x_1)$ with probability at least $4\delta^2$. The construction of $\mA'$ is similar to the case $m'=1$ in the proof of Theorem~\ref{thm:kcvy}. We first execute $\mA_1$ for the first round of the protocol, yeilding $(y,\ket{\psi_y})$. We then create
\[ \ket{\tilde{\psi}_y} \,=\, \frac{1}{\sqrt{2^{2n}}} \sum_{r\in\{0,1\}^{2n}} \ket{r} \ket{\psi_y}\;.\]
We define the following unitary $U$ on this state. $U$ coherently executes $\mA_1$ for the second round of the protocol, treating the first register as the verifier's question, and writes the outcome $d\in\{0,1\}^{n}$ in an ancilla register. 
Then, $U$ coherently executes $\mA_2$ on all registers, writing the outcome $b$ in another ancilla register. 
This unitary satisfies the two conditions of Definition~\ref{def:gl}, for $\eps = \delta$ and the string $a=x_0\| x_1$. 
This is because $\hat{c}_0\cdot \hat{c}_1 = (-1)^{r\cdot (x_0\| x_1)}$, as can be verified from~\eqref{def:hat-c}. 
We now define $\mA'$ to execute the algorithm of Theorem~\ref{thm:qgl}. By Theorem~\ref{thm:qgl}, $\mA'$ returns a claw $(x_0,x_1)$ with probability $4\delta^2$. 
By the claw-free property, it follows that $\delta = \negl(\lambda)$. This proves the theorem.
\end{proof}

Similarly, we apply Theorem~\ref{thm:qubit} to obtain the following consequence.

\begin{corollary}[Qubit test from the simplified KCVY protocol]
Suppose that a quantum prover succeeds with probability $\omega- \eps$ in the protocol from Figure~\ref{fig:protocol}. Let $\{\Pi^m_b\}$ be the projective measurement applied by the prover in the third round of the protocol and $S_m = \Pi^m_0-\Pi^m_1$. Then, on average over the transcript $\trans$ obtained in the first two rounds of the protocol it holds that
\[ \bra{{\psi_\trans}} \{ S_0,S_1\}^2 \ket{{\psi_{\trans}}} \,=\, O(\eps)\;,\]
where $\{S_0,S_1\}=S_0S_1+S_1S_0$ and $\ket{\psi_{\trans}}$ is the state of the prover at the end of the second round.
\end{corollary}

\begin{proof}
We apply Theorem~\ref{thm:qubit}. We already verified that Assumption~\ref{ass:parity} holds, for $s = \negl(\lambda)$, in the proof of Theorem~\ref{thm:simplified}. 
Theorem~\ref{thm:qubit} gives a bound on $\bra{{\psi_\trans}}\{S_0,S_1\}^2\ket{{\psi_\trans}}$. Here the observables $S_m$, $m\in \{0,1\}$, are defined from $\{\Pi^m_b\}$ using the definition of $\hat{c}_m$, see~\eqref{eq:def-s}. However, Since the squared anti-commutator $\{S_0,S_1\}^2$ is invariant under exchanges $S_0\leftarrow - S_0$ or $S_1 \leftarrow -S_1$. Therefore, the bound from Theorem~\ref{thm:qubit} also applies for the simpler definition of $S_m = \Pi^m_0-\Pi^m_1$.
\end{proof}

\subsection{The KLVY protocol}

In~\cite{kalai2022quantum} the authors introduce a general ``compiler'' that takes any $2$-prover nonlocal game and transforms it into a $2$-round test of quantumness.\footnote{Their results apply to $k$-prover nonlocal games; here we only consider the case where $k=2$.} They prove that the resulting protocol has classical soundness equal to the classical value of the nonlocal game, up to an additive term that is negligible in the security parameter $\lambda$, assuming the security of a quantum fully homomorphic encryption scheme with specific properties---namely, that it allows classical encryption of classical messages and that it satisfies a natural ``aux-input correctness'' property which they define. (Both properties are satisfied by the scheme from~\cite{mahadev2020classical}.)

Here we apply our general results to recover classical soundness of the KLVY protocol, when applied to the celebrated nonlocal game CHSH, which is based on the Bell inequality by Clauser et al.~\cite{clauser1969proposed}. Furthermore, we show quantum soundness of the same protocol (the authors were only able to establish quantum completeness; our bound matches theirs) and that the protocol can be used as a test for a qubit.

The KLVY protocol for the CHSH game is described in Figure~\ref{fig:klvy}. For the definition of a quantum fully homomorphic encryption scheme, and the specific properties required here, we refer to~\cite[Definition 2.3]{kalai2022quantum}.

\begin{figure}[!htbp]
  \centering
  \begin{gamespec}
Fix a quantum homomorphic encryption scheme $(\Gen,\Enc,\Eval,\Dec)$ and a security parameter $\lambda$.
    \begin{enumerate}
      \setlength\itemsep{1pt}
    \item  The verifier samples $sk\leftarrow \Gen(1^{\lambda})$. They sample an $x\in\{0,1\}$ uniformly at random and set $\hat{x}\leftarrow\Enc(sk,x)$. They send $\hat{x}$ to the prover. The prover responds with a ciphertext $\hat{a}$.
    \item The verifier sends $m\in\{0,1\}$ chosen uniformly at random to the prover. The prover responds with a bit  $b\in \{0,1\}$.
    \item The verifier computes $a\leftarrow \Dec(sk,\hat{a})$. They accept if and only if $(-1)^b=\hat{c}_m$, where $\hat{c}_m = (-1)^a (-1)^{xm}$.
    \end{enumerate}
  \end{gamespec}
  \caption{The KLVY protocol, specialized to the CHSH game.}
  \label{fig:klvy}
\end{figure}

\begin{theorem}[Classical and quantum soundness of the KLVY protocol for the CHSH game]\label{thm:klvy}
The maximum probability with which a classical (resp. quantum) polynomial-time prover may succeed in the $2$-round protocol from Figure~\ref{fig:protocol} is $\frac{3}{4}+\negl(\lambda)$ (resp. $\cos^2 \frac{\pi}{8} + \negl(\lambda)$).
\end{theorem}

\begin{proof}
Similarly to the proof of Theorem~\ref{thm:simplified}, it suffices to (a) show that the protocol can be reformatted as an instance of our protocol template, and (b) show that Assumption~\ref{ass:parity} holds for some function $s(\lambda,\kappa)$.

Step (a) is very direct: the first round of the protocol is Phase A, and the second round is Phase B. Step (b) is also straightforward. The main observation is that $\hat{c}_0\cdot \hat{c}_1=(-1)^x$. However, $x$ was only given in encrypted form to the prover. So, by semantic security of the homomorphic encryption scheme it should not be able to predict it with any non-negligible advantage.

We proceed with the details. Let $\mA=(\mA_1,\mA_2)$ be a parity adversary. Since there is no test in Phase A of the protocol, we can assume that $\kappa=0$. Suppose that $\mA_2$ has advantage $\delta$ in the parity guessing task. By definition $(\mA_1,\mA_2)$ can be combined into a quantum polynomial-time algorithm that  returns a guess for $x$ that is correct with probability $\frac{1}{2}+\delta$. It follows that $\delta=\negl(\lambda)$, concluding the proof.
\end{proof}

Similarly, we apply Theorem~\ref{thm:qubit} to obtain the following consequence.

\begin{corollary}[Qubit test from the KLVY protocol for the CHSH game]
Suppose that a quantum prover succeeds with probability $\omega- \eps$ in the protocol from Figure~\ref{fig:klvy}. Let $\{\Pi^m_b\}$ be the projective measurement applied by the prover in the second round of the protocol and $S_m = \Pi^m_0-\Pi^m_1$. Then, on average over the transcript $\trans$ obtained in the first round of the protocol it holds that
\[ \bra{{\psi_\trans}} \{ S_0,S_1\}^2 \ket{{\psi_{\trans}}} \,=\, O(\eps)\;,\]
where $\{S_0,S_1\}=S_0S_1+S_1S_0$ and $\ket{\psi_{\trans}}$ is the state of the prover at the end of the first round.
\end{corollary}

\begin{proof}
We apply Theorem~\ref{thm:qubit}. We already verified that Assumption~\ref{ass:parity} holds, for $s = \negl(\lambda)$, in the proof of Theorem~\ref{thm:klvy}. Since $\hat{c}_0=(-1)^a$, the observable $S_0$ as defined in~\eqref{eq:def-s} is the same observable as $S_0$ defined in the corollary. The observable $S_1$ defined in~\eqref{eq:def-s} is $S_1 = 2\Pi^1_{a+x}-\Id$, where $a,x$ are determined by the transcript of the first phase. Theorem~\ref{thm:qubit} gives a bound on $\bra{{\psi_\trans}}\{S_0,S_1\}^2\ket{{\psi_\trans}}$. Since this quantity is invariant under exchange $S_1 \leftarrow -S_1$, the corollary follows for the simpler definition of $S_1 = \Pi^1_0-\Pi^1_1$.
\end{proof}

\appendix

\section{A trigonometric identity} \label{appendix:trig}

\begin{lemma}
  The following inequality holds for all $\alpha, \beta \in [0, 2\pi]$:
  \begin{equation*}
    \cos^2(\alpha) + \cos^2(\beta) \leq \left| 2\cos^2(\alpha - \beta) - 1 \right| + 2 \cos^2(\pi/8)
  \end{equation*}
\end{lemma}

\begin{proof}
  Using $\cos^2(\phi) = \frac{1}{2}(1 + \cos(2 \phi))$ and that $\cos^2(\pi/8) = \frac{1}{2}\left(1 + \frac{1}{\sqrt{2}} \right)$, we can rewrite the inequality as
  \begin{equation*}
    \frac{1}{2}(2 + \cos(2\alpha) + \cos(2\beta)) \leq \left| 2\cos^2(\alpha - \beta) - 1 \right| + 1 + \frac{1}{\sqrt{2}}\;,
  \end{equation*}  
  which after simplification and using the cosine sum rule becomes
  \begin{equation*}
    \cos(\alpha + \beta) \cos(\alpha - \beta) \leq \left| 2\cos^2(\alpha - \beta) - 1 \right| + \frac{1}{\sqrt{2}}\;.
  \end{equation*}    
  Let $x = \alpha + \beta, y = \alpha - \beta$, so that it will suffice to show
  \begin{equation*}
    \cos(x)\cos(y) \leq \left| 2\cos^2(y) - 1 \right| + \frac{1}{\sqrt{2}}\;.
  \end{equation*}
  Note that if $\cos(x)$ and $\cos(y)$ have opposite signs, the inequality is trivially satisfied, as the left-hand side will be non-positive while the right-hand side is always positive. Without loss of generality we restrict to the case where $\cos(x) \geq 0$ and $\cos(y) \geq 0$ (the case where they're both negative is analogous). As $\cos(x) \leq 1$, it's sufficient to show that
  \begin{equation*}
    \cos(y) \leq \left| 2\cos^2(y) - 1 \right| + \frac{1}{\sqrt{2}}\;.
  \end{equation*}
  Taking $t = \cos(y)$, with $0 \leq t \leq 1$, it suffices to show
  \begin{equation*}
    t \leq \left| 2t^2 - 1 \right| + \frac{1}{\sqrt{2}}\;.
  \end{equation*}
  Suppose first that $2t^2 - 1 \geq 0$ which means (since $t \geq 0$) that $t \geq \frac{1}{\sqrt{2}}$. In this case, we have to show that
  \begin{equation*}
    0 \leq 2t^2 - t - 1 + \frac{1}{\sqrt{2}}\;.
  \end{equation*}
  This follows from noting that $2t^2 - t - 1 + \frac{1}{\sqrt{2}}$ has roots $t_1 = \frac{1}{2} - \frac{1}{\sqrt{2}}$ and $t_2 = \frac{1}{\sqrt{2}}$ and is positive for all $t \leq t_1$ and $t \geq t_2$. Since we assumed $t \geq \frac{1}{\sqrt{2}}$ the result follows.

  \noindent Now suppose that $2t^2 - 1 \leq 0$ which means (since $t \geq 0$) that $0 \leq t \leq \frac{1}{\sqrt{2}}$. In this case, we have to show that
  \begin{equation*}
    0 \leq -2t^2 - t + 1 + \frac{1}{\sqrt{2}}\;.
  \end{equation*}
  Here, the roots are $t_1 = -\frac{1}{2} - \frac{1}{\sqrt{2}}$ and $t_2 = \frac{1}{\sqrt{2}}$ and the expression is positive for all $t_1 \leq t \leq t_2$. Since $0 \leq t \leq \frac{1}{\sqrt{2}}$, the inequality is satisfied, concluding the proof.
\end{proof}

\bibliography{refs}

\end{document}